\documentclass{amsart}

\usepackage[a4paper,headheight=0.3cm]{geometry}
\usepackage{fancyhdr}
\usepackage{url}
\usepackage[pdftitle={Finite approximations to coherent choice},%
pdfauthor={Matthias C. M. Troffaes},%
pdfkeywords={decision making; E-admissibility; maximality; numerical analysis; lower prevision; sensitivity analysis}]{hyperref}

\usepackage{amssymb}

\usepackage{enumerate}
\usepackage{amsmath}
\usepackage{amsthm}
\usepackage{amssymb}
\usepackage{tikz}

\newcommand{\pspace}{\Omega}
\newcommand{\gambles}[1][\pspace]{\mathcal{L}(#1)}
\newcommand{\ud}{\,{\rm d}}
\newcommand{\SetR}{\mathbb{R}}

\newcommand{\SetN}{\mathbb{N}}
\DeclareMathOperator{\opt}{opt}
\newcommand{\partition}{\mathcal{A}}

\newcommand{\probabilities}[1][\pspace]{\mathcal{P}(#1)}

\newcommand{\credal}{\mathcal{M}}
\renewcommand{\vec}[1]{\underline{#1}}
\DeclareMathOperator{\cohull}{co}
\DeclareMathOperator{\ext}{ext}
\DeclareMathOperator{\closure}{cl}

\newtheorem{theorem}{Theorem}
\newtheorem{lemma}[theorem]{Lemma}
\newtheorem{corollary}[theorem]{Corollary}

\newtheorem{definition}[theorem]{Definition}

\author{Matthias C. M. Troffaes}
\address{Durham University, Dept. of Mathematical Sciences, Science Laboratories, South Road, Durham DH1 3LE, United Kingdom}
\email{matthias.troffaes@gmail.com}
\title{Finite Approximations To Coherent Choice}

\begin{document}

\begin{abstract}
  This paper studies and bounds the effects of approximating loss
  functions and credal sets on choice functions, under very
  weak assumptions. In particular, the credal set is assumed to be
  neither convex nor closed. The main result is that the effects of
  approximation can be bounded, although in general, approximation of
  the credal set may not always be practically possible. In case of
  pairwise choice, I demonstrate how the situation can be improved by
  showing that only approximations of the extreme points of the closure of the convex
  hull of the credal set need to be taken into account, as expected.
\end{abstract}

\maketitle

\thispagestyle{fancy}

\section{Introduction}

Classical decision theory tells a decision maker to choose that option which
maximises his expected utility. A generalisation of this principle is compelling when the probabilities and utilities
relevant to the problem are not well known. 
Choice functions are one such generalisation,
and select a set of optimal options:
instead of
pointing to a single solution based on possibly wrong assumptions,
choice functions provide a set of optimal options. The decision maker can then
investigate further if the set is too large, or not, if for
instance the optimal set is a singleton, or if a single option from
the set stands out from the rest by other arguments.

However, in modelling decision problems, we often afford ourselves the
luxury of infinite spaces and infinite sets, making those problems sometimes
hard to solve analytically. In such cases we must resort to computers,
and these cannot handle
random variables on infinite spaces, let alone arbitrary infinite sets of probabilities. Hence,
in that case we must approximate our infinite sets by finite ones.
By taking the finite sets sufficiently large,
hopefully the approximation reflects the true result accurately.
This paper confirms this intuition when modelling choice
functions induced by arbitrary (not necessarily convex) sets of
probabilities and a single cardinal utility,
extending similar results known in classical decision theory
\cite{1968:fishburn::sensitivity,1969:pierce::sensitivity:prior}. 

The paper is organised as follows. Section~\ref{sec:choice} introduces
notation, and briefly reviews the theory of coherent choice functions
and their role in decision theory. In
Section~\ref{sec:approx:gambles:probs:prevs} the building blocks for a
theory of approximation are introduced, along with some useful results
on what they imply for loss functions, sets of probabilities, and
expected utility. The main part of the paper begins in
Section~\ref{sec:approx:choice}, studying and bounding the effects of
approximation on coherent choice functions.
Section~\ref{sec:approx:pairwise:choice} improves the results of the
previous section for pairwise choice. Section~\ref{sec:conclusion}
concludes the paper. Some essential but technical results on
approximating the standard simplex in $\SetR^n$ are deferred to an
appendix.

\section{Choice Functions}
\label{sec:choice}

Let $\pspace$ denote an arbitrary set of states. Bounded random
quantities on $\pspace$, i.e.\ bounded maps from $\pspace$ to $\SetR$,
are also called \emph{gambles} \cite{1991:walley}, and will be
denoted by $f$, $g$, \dots\ $\gambles$ denotes the set of all gambles
on $\pspace$. Finitely additive probability measures, or briefly
\emph{probability charges} \cite{1983:rao}, are denoted by $P$,
$Q$, \dots\ and $\probabilities$ denotes the set of all probability
charges on the power set $\wp(\pspace)$ of $\pspace$.

In a decision problem, we desire to choose an optimal option $d$ from
a set $D$ of options. 
Choosing $d$ induces an uncertain reward $r$ from a set $R$ of rewards,
with probability charge $\mu_d(\cdot|w)$ over $\wp(R)$, depending on the outcome of the uncertain
state $w\in\pspace$. 
For each $w\in\pspace$, $\mu_d(\cdot|w)$ is a \emph{lottery} over $R$, and
as a function of $w$, $\mu_d(\cdot|\cdot)\colon w\mapsto \mu_d(\cdot|w)$ is a
\emph{horse lottery} or \emph{act}.

If we model our belief about states and rewards
by a probability charge $P$ on $\wp(\pspace)$ and a state dependent utility function $U(\cdot|w)$ on $R$,
then utility theory \cite{1944:neumann,anscombe:1963,1974:definetti} tells us
to choose a decision $d$ which maximises the expected utility, or prevision:
\begin{align*}
  E(d)
  &=\int_\pspace \left(\int_R U(r|w) \ud \mu_d(r|w)\right) \ud P(w)\\
  &=\int_\pspace f_d(w)\ud P(w)
\end{align*}
where $f_d(w)=\int_R U(r|w) \ud \mu_d(r|w)$ is the gamble associated with
decision $d$, and the integrals are Dunford integrals \cite{1983:rao}. For simplicity, in this paper, we assume $U(r|w)$ to be
bounded, i.e. $$\sup_{r,w}U(r|w)-\inf_{r,w}U(r|w)<+\infty$$
Among other things, this ensures that relative approximation 
can be defined, as in
Section~\ref{sec:approx:gambles:probs:prevs},
without technical complications.

A decision which maximises expected utility is called a \emph{Bayes
  decision} for the decision problem $(\pspace,D,P,U)$.

However, if we are not sure about the probability of all events and the utility of all rewards, a
more reliable design is to use a family $(P_\alpha,U_\alpha)_{\alpha\in\aleph}$ of probability-utility pairs (where $\aleph$ is an arbitrary index set),
and to elicit from $D$ those options which maximise expected
utility with respect to at least one of the pairs $(P_\alpha,U_\alpha)$.
First, for each $\alpha\in\aleph$, let
\begin{equation*}
  E_\alpha(d)
  =\int_\pspace f^\alpha_d(w)\ud P_\alpha(w)
\end{equation*}
where $f^\alpha_d(w)=\int_R U_\alpha(r|w) \ud \mu_d(r|w)$ is the gamble associated with
decision $d$ and model $\alpha\in\aleph$. Then we define:
\begin{definition}\label{def:optimality}
  A decision $d\in D$ is called an \emph{optimal decision} for the decision
  problem $(\pspace,D,(P_\alpha,U_\alpha)_{\alpha\in\aleph})$ if $d$ belongs to the set
  \begin{align*}
    \opt(\pspace,D,(P_\alpha,U_\alpha)_{\alpha\in\aleph})
    &=\left\{
      d\in D\colon
      (\exists\alpha\in\aleph)(\forall e\in D)
      (E_\alpha(d)\ge E_\alpha(e))
    \right\}
    \\
    &=\left\{
      d\in D\colon
      (\exists\alpha\in\aleph)
      \left(E_\alpha(d)=\sup_{e\in D}E_\alpha(e)\right)
    \right\}
  \end{align*}
\end{definition}

As such, the operator $\opt$ selects a \emph{set} of optimal decisions, namely all
decisions which are Bayes with respect to
$(\pspace,D,P_\alpha,U_\alpha)$ for at least one $\alpha\in\aleph$.
Such an operator is called a \emph{choice function} or \emph{optimality
  operator} \cite{2005:decooman:troffaes::dynprog:impgain,2007:troffaes:decision:intro}.

In case
$(P_\alpha,U_\alpha)_{\alpha\in\aleph}=\credal\times\mathcal{U}$ for some
convex sets $\credal$ and $\mathcal{U}$, optimality as defined above
is also called \emph{E-admissibility} \cite[Sec.~4.8]{1980:levi}.

There are many ways to define a choice function starting from a
set $(P_\alpha,U_\alpha)_{\alpha\in\aleph}$ (see
\cite{1980:levi,1995:seidenfeld,1991:walley,2004:seidenfeld::rubinesque,2007:troffaes:decision:intro}).
The one in
Definition~\ref{def:optimality} satisfies an
interesting set of axioms \cite{2004:seidenfeld::rubinesque,2007:seidenfeld::choice:representation}, and is the subject
of a representation theorem in case utility is precise and
state independent (i.e.\ if $U_\alpha(r|w)$ depends neither on
$\alpha$ nor on $w$) and $\pspace$ is finite (for infinite $\pspace$ the representation theorem is subject to additional constraints, which preclude merely finitely additive probabilities over $\pspace$) \cite{2007:seidenfeld::choice:representation}.

For the sake of simplicity, 
we shall only be concerned about
decision problems with precise and state independent utility
functions, i.e.\ when
$(P_\alpha,U_\alpha)_{\alpha\in\aleph}=\credal\times\{U\}$ with $U\colon R\to \SetR$
a bounded state independent utility over $R$ and
\begin{equation*}
  \credal=\{P_\alpha\colon\alpha\in\aleph\}
\end{equation*}
The set $\credal$ is called a \emph{credal set} as it represents our belief about $w\in\pspace$.
We can identify $\credal$ itself as index set, and write
\begin{equation*}
  E_P(d)
  =\int_\pspace f_d(w)\ud P(w)
\end{equation*}
with $f_d(w)=\int_R U(r) \ud \mu_d(r|w)$, for any $P\in\credal$.

Finally, defining the loss function $L\colon D\times\pspace\to\SetR$ as
$L(d,w)=-f_d(w)$, the expected value $E_P(d)$ is uniquely determined by
$P$ and $L$ alone: we need not be concerned explicitly with $R$,
$\mu_d(r|w)$, and $U(r)$.

\section{Approximate Gambles, Probabilities, and Previsions}
\label{sec:approx:gambles:probs:prevs}

Let $\partition=\{A_1,\dots,A_n\}$ denote a finite partition of
$\pspace$. As we approximate $\pspace$ by the finite set $\partition$,
we also need to approximate decisions, gambles, and probability charges on
$\pspace$.

Let $\epsilon\ge 0$.
For a gamble $f$ in $\gambles$ and a gamble $\hat{f}$ in
$\gambles[\partition]$, we shall write $f\sim_\epsilon\hat{f}$
if
\begin{equation*}
  \max_{A\in\partition}
  \sup_{w\in A}
  \left|
    f(w)-\hat{f}(A)
  \right|
  \le[\sup f-\inf f]\epsilon
\end{equation*}
Note that $f\sim_\epsilon\hat{f}$ implies $af+b\sim_\epsilon
a\hat{f}+b$, for any real numbers $a$ and $b$, $a>0$. Therefore, the relation
$\sim_\epsilon$ is invariant with respect to positive linear
transformations of utility: it only depends on our preferences over lotteries, and not on our particular choice of
utility scale.

For a probability charge $P$ in $\probabilities$, and a probability charge
$\hat{P}$ in $\probabilities[\partition]$, we shall write
$P\sim_{\epsilon}\hat{P}$
if
\begin{equation*}
  \sum_{A\in\partition}
  \left|
    P(A)-\hat{P}(A)
  \right|
  \le\epsilon
\end{equation*}
Note that this implies $|P(A)-\hat{P}(A)|\le\epsilon$ for
any $A\in\wp(\partition)$. Also note the differences
between the definitions of $\sim_\epsilon$ for gambles and bounded
charges.

For a loss function $L$ on $D\times\pspace$ and a loss function
$\hat{L}$ on $D\times\partition$ we write
$L\sim_\epsilon\hat{L}$ if for all $d\in D$
\begin{equation*}
  f_d\sim_\epsilon\hat{f}_d
\end{equation*}
(with $f_d(w)=-L(d,w)$ and $\hat{f}_d(A)=-\hat{L}(d,A)$).

For a subset $\credal$ of
$\probabilities$ and a subset $\hat{\credal}$ of
$\probabilities[\partition]$, we write
$\credal\sim_\epsilon\hat{\credal}$ if for every $P$ in $\credal$
there is a $\hat{P}$ in $\hat{\credal}$ such that
$P\sim_\epsilon\hat{P}$, and for every $\hat{P}$ in $\hat{\credal}$ there is a  $P$ in $\credal$ such that
$P\sim_\epsilon\hat{P}$. 

A few useful results about approximations are stated in the next lemmas.

\begin{lemma}\label{lem:approx:gambles}
    Assume that $D$ is finite. Then, for every loss function $L$ on $D\times\pspace$ and every
    $\epsilon>0$, there is a finite partition $\partition$ of
    $\pspace$ and a loss function $\hat{L}$ on $D\times\partition$
    such that
    $L\sim_\epsilon\hat{L}$ and
    $|\partition|\le (1+1/\epsilon)^{|D|}$.
\end{lemma}
\begin{proof}
  Consider any $d$ in $D$, and let $R_d=\sup f_d-\inf f_d$.
  Because $f_d$ is bounded, we can embed the range of $f_d$ in
  $k$ intervals $I_1$, \dots, $I_k$ of length $R_d\epsilon$, say
  \begin{equation*}
    [\inf f_d,\inf f_d+R_d\epsilon),\,
    [\inf f_d+R_d\epsilon,\inf f_d+2 R_d\epsilon),
    \dots, 
    [\inf f_d+(k-1)R_d\epsilon,\inf f_d+kR_d\epsilon)
  \end{equation*}
  with $k$ such that $\sup f_d\in I_k$. Therefore, $\inf
  f_d+(k-1)R_d\epsilon\le\sup f_d<\inf f_d+kR_d\epsilon$ and hence
  $k-1\le 1/\epsilon<k$. Observe that $k$ is independent of
  $d\in D$.

  The sets $A_1$,
  \dots, $A_k$ defined by
  \begin{equation*}
    A_j=f_d^{-1}(I_j)
  \end{equation*}
  form a finite partition $\partition_d=\{A_j\colon
  A_j\neq\emptyset\}$ of cardinality $|\partition_d|\le k\le 1+1/\epsilon$
  and the gamble $\hat{f}_d\in\gambles[\partition_d]$ defined by
  \begin{equation*}
    \hat{f}_d(A_i)=\inf_{w\in A_i}f_d(w)
  \end{equation*}
  satisfies
  \begin{align*}
    \sup_{w\in A_j}
    \left|
      f_d(w)-\hat{f}_d(A_j)
    \right|
    &
    =
    \sup_{f_d(w)\in I_j}
    \left|
      f_d(w)-\inf_{f_d(w)\in I_j}f_d(w)
    \right|
    \\
    &
    \le
    \sup I_j - \inf I_j
    =
    R_d\epsilon
  \end{align*}
  for all $A_j\in\partition_d$; hence $f_d\sim_\epsilon\hat{f}_d$.
  Defining $\hat{L}(d,A)=-\hat{f}_d(A)$ for all $d\in D$, we have
  $L\sim_\epsilon\hat{L}$.

  The finite collection of partitions $\{\partition_d\colon
  d\in D\}$ has a smallest common refinement $\partition$. Since each
  $\partition_d$ has no more than $1+1/\epsilon$ elements,
  $\partition$ has no more than $(1+1/\epsilon)^{|D|}$ elements.
  Indeed, two partitions of cardinalities $k_1$ and $k_2$ respectively have a
  smallest common refinement of cardinality no more than $k_1k_2$. 
  By induction, $n$ partitions of cardinalities $k_1$, \dots,
  $k_n$ have a smallest common refinement of cardinality no more than $\prod_{j=1}^n
  k_j$ and hence,
  \begin{equation*}
    |\partition|\le(1+1/\epsilon)^{|D|}
  \end{equation*}
\end{proof}

Table~\ref{tab:partition} lists upper bounds on the
size of the partition, to ensure $L\sim_\epsilon\hat{L}$, for various
values of $\epsilon$ and $|D|$, according to Lemma~\ref{lem:approx:gambles}.
\begin{table}
  \begin{tabular}{r|rrrrr}
    &      $\epsilon$: \\
    &      $0.2$ & $0.1$ & $0.05$ & $0.02$ & $0.01$ \\
    \hline
    $|D|$:
    $2$  & $1.6$  & $2.1$  & $2.6$  & $3.4$  & $4.0$ \\
    $4$  & $3.1$  & $4.2$  & $5.3$  & $6.8$  & $8.0$ \\
    $8$  & $6.2$  & $8.3$  & $10.6$  & $13.7$  & $16.0$ \\
    $16$  & $12.5$  & $16.7$  & $21.2$  & $27.3$  & $32.1$ \\
    $32$  & $24.9$  & $33.3$  & $42.3$  & $54.6$  & $64.1$
  \end{tabular}
  \caption{Upper bound on $\log_{10}(|\partition|)$, i.e. the logarithm of the cardinality of the finite partition $\partition$ for various values of precision $\epsilon>0$ and number of decisions (see Lemma~\ref{lem:approx:gambles}).}
  \label{tab:partition}
\end{table}

Let $\binom{a}{b}$ be the binomial coefficient, defined for all real
numbers $a\ge b\ge 0$ by
$$\binom{a}{b}=\frac{\Gamma(a+1)}{\Gamma(b+1)\Gamma(a-b+1)}$$ with
$\Gamma$ the Gamma function.

\begin{lemma}\label{lem:approx:probabilities}
  For every subset $\credal$ of $\probabilities$, every $\delta>0$,
  and every finite partition $\partition$ of $\pspace$, there is a
  finite subset $\hat{\credal}$ of $\probabilities[\partition]$ such
  that $\credal\sim_\delta\hat{\credal}$ and
  $|\hat{\credal}|\le\binom{|\partition|(1+1/\delta)}{|\partition|-1}$.
\end{lemma}
\begin{proof}
  Consider any $P$ in
  $\credal$. Let $n=|\partition|$ and let the elements of $\partition$
  be $A_1$, \dots, $A_n$. Consider the vector
  $\vec{x}=(P(A_1),\dots,P(A_n))$ in $\Delta^n$. Let $N$ be the
  smallest natural number such that $N\ge
  n/\delta$.
  
  By Lemma~\ref{lem:simplex:approx} in the appendix, there is a vector $\vec{y}$ in
  $\Delta^n_N$ such that
  \begin{equation*}
    |\vec{x}-\vec{y}|_1< n/N\le\delta
  \end{equation*}
  Define $\hat{P}$ in $\probabilities[\partition]$ by
  \begin{equation*}
    \hat{P}(A_i)=y_i
  \end{equation*}
  for all $i\in\{1,\dots,n\}$---by finite additivity, $\hat{P}$ is
  well defined on $\wp(\partition)$. By construction,
  $P\sim_\delta\hat{P}$ because
  \begin{equation*}
    \sum_{i=1}^n
    \left|
      P(A_i)-\hat{P}(A_i)
    \right|
    =
    |\vec{x}-\vec{y}|_1
    <\delta
  \end{equation*}
  
  Approximating each $P$ in $\credal$ in this manner, the set
  \begin{equation*}
    \hat{\credal}=\{\hat{P}\colon P\in\credal\}
  \end{equation*}
  is finite as each of its elements corresponds to an element of the
  finite set $\Delta^n_N$, and therefore
  $|\hat{\credal}|\le|\Delta^n_N|$.  By
  Lemma~\ref{lem:simplex:cardinality} in the appendix,
  \begin{align*}
    |\hat{\credal}|
    &\le\binom{N+n-1}{N}=\binom{N+n-1}{n-1}
    \\
    &\le\binom{n/\delta+1+n-1}{n-1}
    =\binom{|\partition|(1+1/\delta)}{|\partition|-1}
  \end{align*}
  The second inequality follows from the fact that $\binom{a}{b}$ is
  strictly increasing in $a$, for fixed $b$ (for integer $a$ and $b$
  this follows immediately from Pascal's triangle; the general case
  follows from the properties of the Gamma function).
\end{proof}

Table~\ref{tab:credal} lists upper bounds on the cardinality  of $\hat{\credal}$ on a logarithmic scale,
for some values of $|\mathcal{A}|$ and $\delta$. The
cardinality grows enormously fast with increasing $|\mathcal{A}|$ and $1/\delta$. Within the range of Table~\ref{tab:credal}, an exponential trend is obvious.
The table shows that the influence of $|\mathcal{A}|$ is much larger than the
influence of $\delta$: more precisely, doubling $|\partition|$ increases
$|\hat{\credal}|$ by far more than halving $\delta$.
\begin{table}
  \begin{tabular}{r|rrr}
    &        $\delta$: \\
    &        $0.2$  & $0.1$  & $0.05$ \\
    \hline
    $|\mathcal{A}|$:
    $4$  & $3.3$  & $4.1$  & $5.0$ \\
    $8$  & $7.9$  & $9.8$  & $11.8$ \\
    $12$  & $12.5$  & $15.5$  & $18.7$ \\
    $16$  & $17.1$  & $21.3$  & $25.6$ \\
    $20$  & $21.8$  & $27.1$  & $32.6$ \\
    $24$  & $26.4$  & $32.9$  & $39.5$ \\
    $28$  & $31.1$  & $38.6$  & $46.5$ \\
    $32$  & $35.8$  & $44.4$  & $53.4$ \\
    \hline
    $\log_{10}(|\mathcal{A}|)$:
    $0.7$  & $4.4$  & $5.5$  & $6.7$  \\
    $1.4$  & $27.6$  & $34.3$  & $41.3$ \\
    $2.1$  & $144.6$  & $179.5$  & $215.5$ \\
    $2.8$  & $731.3$  & $906.8$  & $1088.2$ \\
    $3.5$  & $3666.1$  & $4544.7$  & $5452.8$ \\
    $4.2$  & $18341.5$  & $22735.9$  & $27277.5$ \\
    $4.9$  & $91719.7$  & $113693.0$  & $136402.5$
  \end{tabular}
  \caption{Upper bound on $\log_{10}(|\hat{\credal}|)$, i.e.\ the logarithm of the cardinality of the finite set of probability charges $\hat{\credal}$, for various values of precision $\delta>0$ and cardinality of the partition $|\partition|$ (see Lemma~\ref{lem:approx:probabilities}).}
  \label{tab:credal}
\end{table}

Next, we study the effect on the expectation if both gambles and
probabilities are approximated. Let us use the notation $E_P(f)=\int_\pspace f(w)\ud P(w)$.
In the lemma below, assume $0<\epsilon<1/2$.

\begin{lemma}\label{lem:approx:expectation}
  For every finite partition $\partition$ of $\pspace$, every
  $f\in\gambles$, $\hat{f}\in\gambles[\partition]$,
  $P\in\probabilities$, and
  $\hat{P}\in\probabilities[\partition]$, the following implications
  hold. If $f\sim_\epsilon\hat{f}$ and $P\sim_\delta\hat{P}$ then
  \begin{equation*}
    \left|
      E_P(f) - E_{\hat{P}}(\hat{f})
    \right|
    \le
    [\sup f-\inf f](\epsilon+\delta(1+2\epsilon))
  \end{equation*}
  and
  \begin{equation*}
    \left|
      E_P(f) - E_{\hat{P}}(\hat{f})
    \right|
    \le
    [\sup \hat{f}-\inf \hat{f}]\left(\frac{\epsilon}{1-2\epsilon}+\delta\right)
  \end{equation*}
\end{lemma}
\begin{proof}
  Let $R=\sup f-\inf f$, $\hat{R}=\sup\hat{f}-\inf\hat{f}$, and write
  $\inf_A f$ for $\inf_{w\in A}f(w)$ and $\sup_A f$ for $\sup_{w\in
    A}f(w)$. Then
  \begin{align*}
    &\left|
      E_P(f) - E_{\hat{P}}(\hat{f})
    \right|
    =
    \left|
      \sum_{A\in\partition}\left(
        \int_A f\ud P - \hat{f}(A)\hat{P}(A)
      \right)
    \right|
    \\
    \intertext{and since $P(A)\inf_A f\le\int_A f\ud P\le P(A)\sup_A f$, there is an $r_A\in[\inf_A f,\sup_A f]$ such that $P(A)r_A=\int_A f\ud P$, and hence}
    &=
    \left|
      \sum_{A\in\partition}\left(
        r_AP(A) - \hat{f}(A)\hat{P}(A)
      \right)
    \right|
    \\
    \intertext{but, because $|f(w)-\hat{f}(A)|\le R\epsilon$ for all $w\in A$, and $\inf_A f\le r_A\le\sup_A f$, it must also hold that $|r_A-\hat{f}(A)|\le R\epsilon$, so
    $
    \left|
      \sum_{A\in\partition}\left(
        r_AP(A) - \hat{f}(A)P(A)
      \right)
    \right|
    \le
    \sum_{A\in\partition}
    \left|
        r_A - \hat{f}(A)
    \right|
    P(A)
    \le
    \sum_{A\in\partition}
    R\epsilon
    P(A)
    =R\epsilon
    $, whence}
    &\le
    \left|
      \sum_{A\in\partition}\left(
        \hat{f}(A)P(A) - \hat{f}(A)\hat{P}(A)
      \right)
    \right|
    +R\epsilon
    \\
    &=
    \left|
      \sum_{A\in\partition}
        \hat{f}(A)
        \left(
          P(A) - \hat{P}(A)
        \right)
    \right|
    +R\epsilon
    \\
    \intertext{and because $\sum_{A\in\partition}(P(A)-\hat{P}(A))=0$,}
    &=
    \left|
      \sum_{A\in\partition}
        (\hat{f}(A)-\inf \hat{f})
        \left(
          P(A) - \hat{P}(A)
        \right)
    \right|
    +R\epsilon
    \\
    &\le
    \sum_{A\in\partition}
    (\hat{f}(A)-\inf \hat{f})
    \left|
      P(A) - \hat{P}(A)
    \right|
    +R\epsilon
    \\
    &\le
    (\sup\hat{f}-\inf \hat{f})
    \sum_{A\in\partition}
    \left|
      P(A) - \hat{P}(A)
    \right|
    +R\epsilon
    \\
    &\le
    \hat{R}\delta+R\epsilon
    \\
    \intertext{and since $R(1+2\epsilon)\ge\hat{R}\ge R(1-2\epsilon)$}
    &\le
    \begin{cases}
      R(1+2\epsilon)\delta+R\epsilon=R(\epsilon+\delta(1+2\epsilon)) \\
      \hat{R}\delta+\hat{R}\epsilon/(1-2\epsilon)=\hat{R}\left(\epsilon/(1-2\epsilon)+\delta\right)
    \end{cases}
  \end{align*}
\end{proof}

Let us now investigate what is the most optimal choice for
$\epsilon>0$ and $\delta>0$. The cardinality of
$\hat{\credal}$ is of largest concern as it grows enormously fast with
increasing cardinality of the finite partition $\partition$ and with increasing
precision $1/\delta$ (see
Table~\ref{tab:credal}). Therefore, as a first step, let us see how we
can minimise $|\hat{\credal}|$, assuming a fixed relative error
$\epsilon+\delta$ on the expectation (see
Lemma~\ref{lem:approx:expectation})---omitting higher
order terms in $\epsilon$ and $\delta$ to simplify the analysis.

We wish to
minimise the upper bound (neglecting lower order terms)
\begin{equation*}
  \binom{(1/(\epsilon^{|D|}\delta)}{1/{\epsilon}^{|D|}}
\end{equation*}
on $|\hat{\credal}|$ along the $\epsilon$--$\delta$-curve
$\gamma(\epsilon,\delta)=\epsilon+\delta=\gamma_*$. Figure~\ref{fig:gamma} demonstrates a typical case: the $\epsilon$--$\delta$-ratio
has a large impact on the upper bound of $|\hat{\credal}|$. In particular, the
curve grows extremely large for small $\epsilon$, because a small $\epsilon$
corresponds to a large partition $\partition$, and the cardinality of the partition has a huge impact on the cardinality of $\credal$ as shown in Table~\ref{tab:credal}.
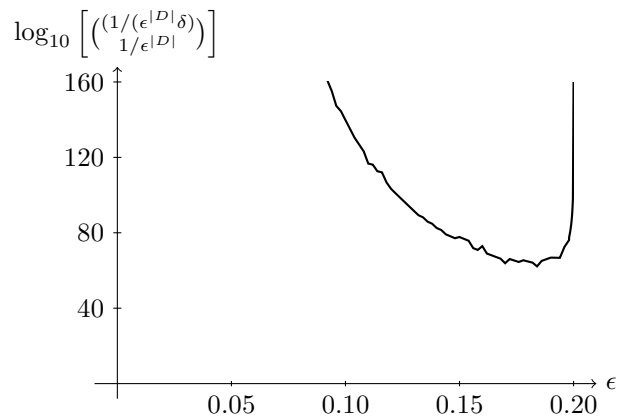
\begin{figure}
  \begin{center}
   \begin{tikzpicture}[xscale=1.5]
     \draw[->] (-0.2,0) -- (4.2,0) node[right] {$\epsilon$};
     \draw[->] (0,-0.2) -- (0,4.2) node[above] {$\log_{10}\left[\binom{(1/(\epsilon^{|D|}\delta)}{1/{\epsilon}^{|D|}}\right]$};
     \foreach \x/\xtext in {1/0.05,2/0.10,3/0.15,4/0.20}
     \draw (\x cm,1pt) -- (\x cm,-1pt) node[anchor=north] {$\xtext$};
     \foreach \y/\ytext in {1/40,2/80,3/120,4/160}
     \draw (1pt,\y cm) -- (-1pt,\y cm) node[anchor=east] {$\ytext$};
     \clip (0,0) rectangle (4,4);
     \draw[xscale=20,yscale=0.025,thick] plot coordinates {
 ( 0.19999 , 160 )
 ( 0.1997 , 97.8608741706 )
 ( 0.1996 , 94.7477940539 )
 ( 0.1995 , 92.3353788655 )
 ( 0.1994 , 90.3662593767 )
 ( 0.19938 , 90.0123381674 )
 ( 0.19936 , 89.6697204213 )
 ( 0.19934 , 89.3376067698 )
 ( 0.19932 , 89.0155709221 )
 ( 0.19926 , 88.1036692063 )
 ( 0.19924 , 87.8162007617 )
 ( 0.19922 , 87.536391416 )
 ( 0.19918 , 86.9973580326 )
 ( 0.19916 , 86.7377472816 )
 ( 0.19914 , 86.4843790152 )
 ( 0.19912 , 86.2370906627 )
 ( 0.1991 , 85.9948947146 )
 ( 0.19904 , 85.3004010684 )
 ( 0.19902 , 85.0787847299 )
 ( 0.199 , 84.8615876423 )
 ( 0.19896 , 84.4396765194 )
 ( 0.19894 , 84.2351765566 )
 ( 0.19892 , 84.0341809342 )
 ( 0.1989 , 83.836960546 )
 ( 0.19886 , 83.4534960747 )
 ( 0.19884 , 83.2667809825 )
 ( 0.19878 , 82.7253836686 )
 ( 0.1987 , 82.0440599959 )
 ( 0.19868 , 81.8805686246 )
 ( 0.19866 , 81.7191751074 )
 ( 0.19864 , 81.5605945889 )
 ( 0.19858 , 81.097875651 )
 ( 0.19854 , 80.8005978141 )
 ( 0.19848 , 80.3694605846 )
 ( 0.1984 , 79.8210394485 )
 ( 0.19838 , 79.6881565118 )
 ( 0.19836 , 79.5571333862 )
 ( 0.19834 , 79.4273484942 )
 ( 0.19832 , 79.2995839298 )
 ( 0.1983 , 79.1732039206 )
 ( 0.19828 , 79.048275089 )
 ( 0.19826 , 78.9248652255 )
 ( 0.19822 , 78.6821190827 )
 ( 0.1982 , 78.5629069535 )
 ( 0.19816 , 78.3283371166 )
 ( 0.19814 , 78.2130924654 )
 ( 0.19812 , 78.0990174872 )
 ( 0.19808 , 77.8745783056 )
 ( 0.19806 , 77.7643168843 )
 ( 0.19802 , 76.1321555741 )
 ( 0.196 , 72.5223007109 )
 ( 0.194 , 66.7324780368 )
 ( 0.19 , 66.8458308753 )
 ( 0.186 , 65.1003605391 )
 ( 0.184 , 62.2010471925 )
 ( 0.182 , 64.1990129644 )
 ( 0.18 , 64.7925715525 )
 ( 0.178 , 65.4478669545 )
 ( 0.176 , 64.537520035 )
 ( 0.174 , 65.3047198335 )
 ( 0.172 , 66.1041335672 )
 ( 0.17 , 63.8745781092 )
 ( 0.168 , 66.2950422457 )
 ( 0.166 , 67.1840668137 )
 ( 0.164 , 68.0996175462 )
 ( 0.162 , 69.0226530496 )
 ( 0.16 , 72.9347036435 )
 ( 0.158 , 70.910201652 )
 ( 0.156 , 71.8728728581 )
 ( 0.154 , 75.7997747596 )
 ( 0.152 , 76.7849575813 )
 ( 0.15 , 77.7840670041 )
 ( 0.148 , 77.1391449471 )
 ( 0.146 , 78.171115114 )
 ( 0.144 , 79.2065097354 )
 ( 0.142 , 81.4806279169 )
 ( 0.14 , 82.5449799151 )
 ( 0.138 , 84.8252664221 )
 ( 0.136 , 85.9354377602 )
 ( 0.134 , 88.2129109128 )
 ( 0.132 , 89.3535622746 )
 ( 0.13 , 91.6421137107 )
 ( 0.128 , 93.9663075186 )
 ( 0.126 , 96.2663052382 )
 ( 0.124 , 98.6116177056 )
 ( 0.122 , 100.938142972 )
 ( 0.12 , 103.318601407 )
 ( 0.118 , 106.770155561 )
 ( 0.116 , 112.099518012 )
 ( 0.114 , 112.666354116 )
 ( 0.112 , 116.15441194 )
 ( 0.11 , 116.737195536 )
 ( 0.108 , 123.227695217 )
 ( 0.106 , 126.822862353 )
 ( 0.104 , 130.448244241 )
 ( 0.102 , 135.116219626 )
 ( 0.1 , 139.805164384 )
 ( 0.098 , 144.566504162 )
 ( 0.096 , 147.322080859 )
 ( 0.094 , 155.211344399 )
 ( 0.092 , 161.089293471 )
 ( 0.09 , 167.048607006 )
 ( 0.088 , 174.047058463 )
 };
   \end{tikzpicture}
  \caption{Upper bound on $\log_{10}|\hat{\credal}|$ for various values of $\epsilon$, with $\epsilon+\delta=0.2$ and $|D|=2$.}
  \label{fig:gamma}
  \end{center}
\end{figure}

\section{Approximate Choice}
\label{sec:approx:choice}

Let us now consider again the decision problem $(\pspace,D,\credal,L)$
with state space $\pspace$, decision space $D$, credal set $\credal$,
and loss function $L$, and reflect upon how the results in the
previous section could be of use in finding the optimal decisions
$\opt(\pspace,D,\credal,L)$.
Can we still find the optimal decisions after approximating the
loss function $L$ and the set of probabilities $\credal$?

As we admit a relative error on gambles and probabilities, and therefore
also on previsions, we should admit a relative error on the choice function
as well.  Let $R_D$ be defined by (recall that $f_d(w)=-L(d,w)$)
\begin{equation*}
  R_D=\sup_{d\in D}[\sup f_d-\inf f_d]
\end{equation*}

\begin{definition}
  Let $\epsilon\ge 0$. A decision $d$ in $D$ is called an
  \emph{$\epsilon$-optimal decision} for the decision problem
  $(\pspace,D,\credal,L)$ if it belongs to the set
  \begin{equation*}
    \opt^\epsilon(\pspace,D,\credal,L)=
    \left\{
      d\in D\colon 
      (\exists P\in\credal)
      \left(
        \sup_{e\in D}E_P(e)-E_P(d)
        \le\epsilon R_D
      \right)
    \right\}
  \end{equation*}
\end{definition}
Note that
\begin{equation*}
\opt^\epsilon(\pspace,D,\credal,aL+b)
=\opt^\epsilon(\pspace,D,\credal,L)
\end{equation*}
for any real numbers $a$ and $b$, $a>0$. In other words,
$\opt^\epsilon(\pspace,D,\credal,L)$ is invariant with respect to positive linear
transformations of utility: $\epsilon$-optimality does not depend on our choice of
utility scale.

Clearly,
\begin{equation*}
  \opt(\pspace,D,\credal,L)\subseteq\opt^\epsilon(\pspace,D,\credal,L)
\end{equation*}
because $$\opt^{\epsilon}(\pspace,D,\credal,L)\subseteq\opt^{\delta}(\pspace,D,\credal,L)$$ whenever $\epsilon\le\delta$, and $$\opt^0(\pspace,D,\credal,L)=\opt(\pspace,D,\credal,L)$$

In approximating a decision problem $(\pspace,D,\credal,L)$, we
start with a finite partition $\mathcal{A}$,
consider a (possibly
finite) set $\hat{\credal}$ such that
$\credal\sim_\delta\hat{\credal}$, and
approximate the loss $L(d,w)$ by a loss $\hat{L}(d,A)$
such that $L\sim_\epsilon\hat{L}$.

\begin{theorem}\label{thm:approx:opt}
  Consider two decision problems $(\pspace,D,\credal,L)$ and
  $(\partition,D,\hat{\credal},\hat{L})$.
  If $L\sim_\epsilon\hat{L}$ and
  $\credal\sim_\delta\hat{\credal}$ then, for any $\gamma\ge 0$,
  \begin{equation}
  \label{eq:thm:approx:opt:1}
    \opt^\gamma(\pspace,D,\credal,L)
    \subseteq
    \opt^{\frac{\gamma}{1-2\epsilon}+2(\frac{\epsilon}{1-2\epsilon}+\delta)}%
    (\partition,D,\hat{\credal},\hat{L})
  \end{equation}
  and
  \begin{equation}
    \label{eq:thm:approx:opt:2}
    \opt^\gamma(\partition,D,\hat{\credal},\hat{L})
    \subseteq
    \opt^{\gamma(1+2\epsilon)+2(\epsilon+\delta(1+2\epsilon))}(\pspace,D,\credal,L)
  \end{equation}
\end{theorem}
\begin{proof}
  We prove Eq.~\eqref{eq:thm:approx:opt:1}.
  Let $d\in\opt^\gamma(\pspace,D,\credal,L)$. Then
  \begin{equation}
    \label{eq:thm:approx:optimality:1}
    \sup_{e\in D}E_P(f_e)-E_P(f_d)\le\gamma R_D
  \end{equation}
  for some $P\in\credal$. Let $\hat{P}$ be such that
  $P\sim_\delta\hat{P}$.
  Because, by Lemma~\ref{lem:approx:expectation},
  \begin{align}
    \nonumber
    \left|
      \sup_{e\in D}E_{\hat{P}}(\hat{f}_e)
      -\sup_{e'\in D}E_P(f_{e'})
    \right|
    & \le
    \nonumber
    \sup_{e\in D}
    \left|
      E_{\hat{P}}(\hat{f}_e)-E_P(f_e)
    \right|
    \\ & \le
    \nonumber
    \sup_{e\in D}[\sup\hat{f}_e-\inf\hat{f}_e](\epsilon/(1-2\epsilon)+\delta)
    \\ & =(\epsilon/(1-2\epsilon)+\delta) \hat{R}_D
    \label{eq:thm:approx:optimality:2}
  \end{align}
  it follows that
  \begin{align*}
    \sup_{e\in D}E_{\hat{P}}(\hat{f}_e)-E_{\hat{P}}(\hat{f}_d)
    &\le
    \sup_{e\in D}E_P(f_e)-E_{\hat{P}}(\hat{f}_d)+(\epsilon/(1-2\epsilon)+\delta) \hat{R}_D
    \\
    \intertext{and again by Lemma~\ref{lem:approx:expectation},}
    &
    \le
    \sup_{e\in D}E_P(f_e)-E_{P}(f_d)+2(\epsilon/(1-2\epsilon)+\delta)\hat{R}_D
    \\
    \intertext{and by Eq.~\eqref{eq:thm:approx:optimality:1},}
    &
    \le
    \gamma R_D+2(\epsilon/(1-2\epsilon)+\delta)\hat{R}_D
    \\
    &
    \le
    [\gamma/(1-2\epsilon)+2(\epsilon/(1-2\epsilon)+\delta)]\hat{R}_D
  \end{align*}
  hence, $d\in\opt^{\gamma/(1-2\epsilon)+2(\epsilon/(1-2\epsilon)+\delta)}(\partition,D,\hat{\credal},\hat{L})$.

  Next, we prove Eq.~\eqref{eq:thm:approx:opt:2}.
  Let $d\in\opt^{\gamma}(\partition,D,\hat{\credal},\hat{L})$. Then
  \begin{equation}\label{eq:thm:approx:optimality:3}
    \sup_{e\in D}E_{\hat{P}}(\hat{f}_e)-E_{\hat{P}}(\hat{f}_d)
    \le
    \gamma \hat{R}_D
  \end{equation}
  Because, by Lemma~\ref{lem:approx:expectation},
  \begin{align}
    \left|
      \sup_{e\in D}E_{\hat{P}}(\hat{f}_e)
      -\sup_{e'\in D}E_P(f_{e'})
    \right|
    & \le
    \nonumber
    \sup_{e\in D}
    \left|
      E_{\hat{P}}(\hat{f}_e)-E_P(f_e)
    \right|
    \\ & \le
    \nonumber
    \sup_{e\in D}[\sup f_e-\inf f_e](\epsilon+\delta(1+2\epsilon))
    \\ & =
    (\epsilon+\delta(1+2\epsilon)) R_D
    \label{eq:thm:approx:optimality:4}
  \end{align}
  we have that
  \begin{align*}
    \sup_{e\in D}E_P(f_e)-E_P(f)
    &
    \le
    \sup_{e\in D}E_{\hat{P}}(\hat{f}_e)-E_P(f)+(\epsilon+\delta(1+2\epsilon)) R_D
    \\
    \intertext{and again by Lemma~\ref{lem:approx:expectation},}
    &
    \le
    \sup_{e\in D}E_{\hat{P}}(\hat{f}_e)
    -E_{\hat{P}}(\hat{f}_e)
    +2(\epsilon+\delta(1+2\epsilon)) R_D
    \\
    \intertext{and by Eq.~\eqref{eq:thm:approx:optimality:3}}
    &
    \le
    \gamma \hat{R}_D
    +2(\epsilon+\delta(1+2\epsilon))R_D
    \\
    &
    \le
    [\gamma(1+2\epsilon)+2(\epsilon+\delta(1+2\epsilon))] R_D
  \end{align*}
  so $d\in\opt^{\gamma(1+2\epsilon)+2(\epsilon+\delta(1+2\epsilon))}(\pspace,D,\credal,L)$.
\end{proof}

If we ignore higher order terms in $\gamma$, $\epsilon$, and $\delta$,
then the above theorem says that when moving from an original decision
problem to an approximate decision problem, or the other way around,
with relative error $\epsilon$ in gambles and relative error $\delta$
in probabilities, the relative error in optimality increases by
$2(\epsilon+\delta)$.  For example, for small $\epsilon$ and $\delta$
the following holds, up to a small error: if $L\sim_\epsilon\hat{L}$ and $\credal\sim_\delta\hat{\credal}$, then
\begin{equation*}
    \opt(\pspace,D,\credal,L)
    \subseteq
    \opt^{2(\epsilon+\delta)}(\partition,D,\hat{\credal},\hat{L})
    \subseteq
    \opt^{4(\epsilon+\delta)}(\pspace,D,\credal,L)
\end{equation*}
So, the approximate problem with relative error $2(\epsilon+\delta)$
will contain all solutions to the original problem with no
relative error, and will, so to speak, not contain any solutions to the original
problem with relative error over $4(\epsilon+\delta)$. Because of this
property,
$\opt^{2(\epsilon+\delta)}(\partition,D,\hat{\credal},\hat{L})$ seems
a logical choice when solving decision problems in practice.

\section{Pairwise Choice}
\label{sec:approx:pairwise:choice}

Table~\ref{tab:credal} reveals that the size of the credal
set is a serious computational bottleneck. Therefore, it is worth investigating
how the size of $\hat{\credal}$ can be reduced, without
compromising the accuracy $\delta>0$. One way to this end is to restrict to
pairwise comparisons, i.e.\ using maximality (see Walley \cite[Sec.~3.7--3.9]{1991:walley}).

\subsection{Maximality}

\begin{definition}
  A decision $d\in D$ is called a \emph{maximal decision} for the decision
  problem $(\pspace,D,\credal,L)$ if $d$ belongs to the set
  \begin{equation*}
    \max(\pspace,D,\credal,L)
    =\left\{
      d\in D\colon
      (\forall e\in D)(\exists P\in\credal)
      \left(E_P(d)\ge E_P(e)\right)
    \right\}
  \end{equation*}
\end{definition}

Denote by $\cohull(\credal)$ the convex hull of $\credal$.
Obviously it holds that
$$\max(\pspace,D,\credal,L)=\max(\pspace,D,\cohull(\credal),L)$$ because
for any $\lambda\in[0,1]$ and any two $P$ and $Q$ in $\credal$,
the inequalities $E_P(d)\ge E_P(e)$ and $E_Q(d)\ge
E_Q(e)$ imply the inequality
$$E_{\lambda P+(1-\lambda) Q}(d)\ge E_{\lambda P+(1-\lambda) Q}(e)$$
This does not hold for optimality as defined in Definition~\ref{def:optimality}: assuming $\pspace$ finite, for any two distinct subsets
$\credal$ and $\credal'$ of $\probabilities$, we can always find a
set $D$ and a loss function $L$ such that
$\opt(\pspace,D,\credal,L)\neq\opt(\pspace,D,\credal',L)$ (see Kadane,
Schervish, and Seidenfeld \cite[Thm.~1,
p.~53]{2004:seidenfeld::rubinesque}).

To understand why the above notion of optimality is called maximality,
consider the strict partial ordering $>$ on $D$ defined by
\begin{equation*}
  e>d \iff (\forall P\in\credal)
      \left(E_P(e)>E_P(d)\right)
\end{equation*}
for any $d$ and $e$ in $D$, that is, $e$ is strictly preferred to $d$
if $e$ is strictly preferred to $d$ with respect to every
$P\in\credal$.
Then,
\begin{equation*}
  \max(\pspace,D,\credal,L)
  =\left\{
    d\in D\colon
    (\forall e\in D)(e\not> d)
  \right\}
\end{equation*}
so $\max(\pspace,D,\credal,L)$ elects those decisions $d$ which are undominated with respect to $>$. Therefore, maximality can be expressed through pairwise
preferences only---again in contrast to
$\opt(\pspace,D,\credal,L)$ as for instance demonstrated by Kadane, Schervish, and
Seidenfeld \cite[Sec.~4, p.~51]{2004:seidenfeld::rubinesque}.

However, because
\begin{equation*}
  \opt(\pspace,D,\credal,L)
  \subseteq
  \max(\pspace,D,\credal,L)
\end{equation*}
we may interpret $\max(\pspace,D,\credal,L)$ as an approximation to
$\opt(\pspace,D,\credal,L)$, an approximation which discards all
preferences but the pairwise ones.

Let us admit a relative error on the choice function $\max$
as well.  Recall, $R_D=\sup_{d\in D}[\sup f_d-\inf f_d]$.

\begin{definition}
  Let $\epsilon\ge 0$. A decision $d$ in $D$ is called an
  \emph{$\epsilon$-maximal decision} for the decision problem
  $(\pspace,D,\credal,L)$ if it belongs to the set
  \begin{equation*}
    {\max}^\epsilon(\pspace,D,\credal,L)
    =
    \{
      d\in D\colon 
      (\forall e\in D)(\exists P\in\credal)
      (
        E_P(e)-E_P(d)
        \le\epsilon R_D
      )
    \}
  \end{equation*}
\end{definition}

\subsection{Approximating Extreme Points}

It turns out that we can restrict our attention to the extreme points
of the closure of the convex hull of $\credal$, with respect to the
topology of pointwise convergence on members of $\gambles$. This topology is characterised by the following notion of convergence: for every directed set $(A,\le)$ and every net $(P_\alpha)_{\alpha\in A}$, we have that
$\lim_\alpha P_\alpha=P$ if
\begin{equation*}
  \lim_\alpha E_{P_\alpha}(f)=E_P(f)\text{ for all }f\in\gambles
\end{equation*}
Without further mention, I will assume this topology on $\probabilities$. See
for instance \cite{1997:schechter} for more information regarding
nets \cite[Chapter~7]{1997:schechter} and
this topology \cite[\S 28.15]{1997:schechter}.

There is a nice
connection between the closure of $\credal$, denoted by $\closure(\credal)$, and $\epsilon$-optimality and $\epsilon$-maximality.

\begin{lemma}\label{lem:credal:closure}
  Assume that $R_D>0$. Let $\epsilon\ge 0$. For any decision problem
  $(\pspace,D,\credal,L)$, the following equality holds:
  \begin{equation}
    \label{eq:lem:credal:closure:max}
    {\max}^{\epsilon}(\pspace,D,\closure(\credal),L)
    =
    \bigcap_{\delta>0}{\max}^{\epsilon+\delta}(\pspace,D,\credal,L)
  \end{equation}
  and if additionally $D$ is finite, then the following equality holds as well:
  \begin{equation}
    \label{eq:lem:credal:closure:opt}
    {\opt}^{\epsilon}(\pspace,D,\closure(\credal),L)
    =
    \bigcap_{\delta>0}{\opt}^{\epsilon+\delta}(\pspace,D,\credal,L)
  \end{equation}
\end{lemma}
\begin{proof}
  We start with proving Eq.~\eqref{eq:lem:credal:closure:max}.

  Assume $d\in{\max}^{\epsilon}(\pspace,D,\closure(\credal),L)$. Consider any
  $e\in D$. By assumption, there is a $P\in\closure(\credal)$ such
  that $E_P(e)-E_P(d)\le R_D\epsilon$. Because $P\in\closure(\credal)$, there is a net
  $(P_\alpha\in\credal)_{\alpha\in A}$ such that $\lim_\alpha E_{P_\alpha}(f)=E_P(f)$
  for all gambles $f$. It follows that $\lim_\alpha E_{P_\alpha}(e)-\lim_\alpha
  E_{P_\alpha}(d)\le R_D\epsilon$. This implies that for every $\delta>0$, there is an
  $\alpha\in A$ such that $E_{P_\alpha}(e)-E_{P_\alpha}(f)\le(\epsilon+\delta) R_D$. So, for every $\delta>0$, there is a $P\in\credal$ such
  that $E_{P}(e)-E_{P}(f)\le(\epsilon+\delta)R_D$. Whence, because this holds for any $e\in D$,
  $d\in{\max}^{\epsilon+\delta}(\pspace,D,\credal,L)$ for all $\delta>0$,
  and therefore, $d\in\bigcap_{\delta>0}{\max}^{\epsilon+\delta}(\pspace,D,\credal,L)$.

  Conversely, assume $d\in\bigcap_{\delta>
    0}{\max}^{\epsilon+\delta}(\pspace,D,\credal,L)$. Consider any $e\in D$.
  Then, for all $\delta>0$, there is a $P_\delta\in\credal$ such
  that $E_{P_\delta}(e)-E_{P_\delta}(f)\le(\epsilon+\delta) R_D$. Hence,
  for all $n\in\SetN$, there is a $P_n\in\credal$ such that
  \begin{equation}
    \label{eq:lem:credal:closure:1}
    E_{P_n}(e)-E_{P_n}(d)\le 1/n+\epsilon R_D
  \end{equation} 
  For any $m\in\SetN$, consider the
  following closed subset of $\probabilities$:
  \begin{equation*}
    \mathcal{R}_m=\closure(\{P_n\colon n\ge m\})
  \end{equation*}
  The collection $\{\mathcal{R}_m\colon m\in\SetN\}$
  satisfies the finite intersection property. By the
  Banach-Alaoglu-Bourbaki theorem \cite[\S 28.29(UF26)]{1997:schechter}
  $\probabilities$ is compact, and hence
  \begin{equation*}
    \mathcal{R}=\cap_{m\in\SetN}\mathcal{R}_m
  \end{equation*}
  is non-empty as well \cite[\S 17.2]{1997:schechter}.

  Take any $R\in\mathcal{R}$. Since each $P_n\in\credal$, it follows
  that each $\mathcal{R}_m\subseteq\closure(\credal)$, and hence
  $R\in\closure(\credal)$. If we can show that $E_R(e)-E_R(d)\le\epsilon R_D$,
  then $d\in{\max}^{\epsilon}(\pspace,D,\closure(\credal),L)$ is established.

  Indeed, fix $m\in\SetN$. Because $R\in\mathcal{R}_m$, there is a
  net $(P_{n_\alpha})_{\alpha\in A}$ in $\{P_n\colon n\ge m\}$---so
  $n_\alpha\ge m$, but $n_\alpha$ is not necessarily an increasing function of
  $\alpha$---such that $\lim_\alpha E_{P_{n_\alpha}}(f_e-f_d)=E_R(f_e-f_d)$.  Hence, for
  each $\gamma>0$, there is an $\alpha\in A$ such that $E_R(e)-E_R(d)\le
  E_{P_{n_\alpha}}(e)-E_{P_{n_\alpha}}(d)+\gamma$, and therefore by Eq.~\eqref{eq:lem:credal:closure:1}, $E_R(e)-E_R(d)\le
  1/n_\alpha+\epsilon R_D + \gamma$. Because this inequality holds for every $m$ and every
  $\gamma>0$, and $n_\alpha\ge m$, it follows that $E_R(e)-E_R(d)\le\epsilon R_D$.

  Let us now prove Eq.~\eqref{eq:lem:credal:closure:opt}, under the
  additional assumption that $D$ is finite. The proof goes along
  similar lines as the one for Eq.~\eqref{eq:lem:credal:closure:max}.

  Assume $d\in{\opt}^{\epsilon}(\pspace,D,\closure(\credal),L)$.  By
  assumption, there is a $P\in\closure(\credal)$ such that
  $E_P(e)-E_P(d)\le R_D\epsilon$ for every $e\in D$. Because
  $P\in\closure(\credal)$, there is a net
  $(P_\alpha\in\credal)_{\alpha\in A}$ such that $\lim_\alpha E_{P_\alpha}(f)=E_P(f)$
  for all gambles $f$. In particular, there is a net
  $(P_\alpha\in\credal)_{\alpha\in A}$ such that $\lim_\alpha E_{P_\alpha}(e)-\lim_\alpha
  E_{P_\alpha}(d)\le R_D\epsilon$ for every $e\in D$.  So, for every $e\in
  D$ and $\delta>0$, there is an $\alpha_{e,\delta}\in A$ such that
  $E_{P_\alpha}(e)-E_{P_\alpha}(f)\le(\epsilon+\delta) R_D$ for all $\alpha\ge
  \alpha_{e,\delta}$. Because $D$ is finite, there is an $\alpha_\delta$ such that
  $\alpha_\delta\ge \alpha_{e,\delta}$ for all $e\in D$. Hence, for every $\delta>0$,
  there is a $\alpha_\delta\in A$ such that
  $E_{P_{\alpha_\delta}}(e)-E_{P_{\alpha_\delta}}(f)\le(\epsilon+\delta) R_D$
  for every $e\in D$. Whence, because $P_{\alpha_\delta}\in\credal$, it follows that
  $d\in{\opt}^{\epsilon+\delta}(\pspace,D,\credal,L)$ for all $\delta>0$,
  and therefore, $d\in\bigcap_{\delta>0}{\opt}^{\epsilon+\delta}(\pspace,D,\credal,L)$.

  Conversely, assume $d\in\bigcap_{\delta>
    0}{\opt}^{\epsilon+\delta}(\pspace,D,\credal,L)$.
  Then, for all $\delta>0$, there is a $P_\delta\in\credal$ such
  that $E_{P_\delta}(e)-E_{P_\delta}(f)\le(\epsilon+\delta) R_D$ for every $e\in D$. Hence,
  for all $n\in\SetN$, there is a $P_n\in\credal$ such that for every $e\in D$
  \begin{equation}
    \label{eq:lem:credal:closure:2}
    E_{P_n}(e)-E_{P_n}(d)\le 1/n+\epsilon R_D
  \end{equation} 
  Now choose any $R$ in
  \begin{equation*}
    \mathcal{R}=\cap_{m\in\SetN}\closure(\{P_n\colon n\ge m\})
  \end{equation*}
  Similarly as before, it can be established that $\mathcal{R}$ is non-empty
  and that $R\in\closure(\credal)$. If we can show that
  $E_R(e)-E_R(d)\le\epsilon R_D$ for all $e\in D$, then
  $d$ indeed belongs to ${\opt}^{\epsilon}(\pspace,D,\closure(\credal),L)$ and
  the desired result is established.

  Indeed, because $R\in\closure(\{P_n\colon n\ge m\})$, for every $e\in D$, there is a
  net $(P_{n_{\alpha,e}})_{\alpha\in A}$ in $\{P_n\colon n\ge m\}$---so
  $n_{\alpha,e}\ge m$---such that $\lim_\alpha E_{P_{n_{\alpha,e}}}(f_e-f_d)=E_R(f_e-f_d)$.  Hence, for every $e\in D$ and every
  $\gamma>0$, there is an $\alpha\in A$ such that $E_R(e)-E_R(d)\le
  E_{P_{n_{\alpha,e}}}(e)-E_{P_{n_{\alpha,e}}}(d)+\gamma$, and therefore by Eq.~\eqref{eq:lem:credal:closure:2}, $E_R(e)-E_R(d)\le
  1/n_{\alpha,e}+\epsilon R_D + \gamma$. Because this inequality holds for every $m$ and every
  $\gamma>0$, and $n_{\alpha,e}\ge m$, it follows that $E_R(e)-E_R(d)\le\epsilon R_D$ for every $e\in D$.
\end{proof}

In particular, assuming $R_D>0$, if for any $\delta>\epsilon>0$
\begin{equation*}
    {\max}^{\epsilon}(\pspace,D,\credal,L)
    =
    {\max}^{\delta}(\pspace,D,\credal,L)
\end{equation*}
then
\begin{equation*}
  {\max}^{\epsilon}(\pspace,D,\credal,L)
  =
  {\max}^{\epsilon}(\pspace,D,\closure(\credal),L)
\end{equation*}
A similar result holds for the ${\opt}^{\epsilon}$ operator for finite $D$.

As a special case, Lemma~\ref{lem:credal:closure} implies an
interesting connection between maximality and $\epsilon$-maximality:
\begin{corollary}\label{cor:credal:closure:maximality}
  Assume that $R_D>0$. For any decision problem
  $(\pspace,D,\credal,L)$, the following equality holds:
  \begin{equation*}
    {\max}(\pspace,D,\closure(\credal),L)
    =
    \bigcap_{\epsilon>0}{\max}^{\epsilon}(\pspace,D,\credal,L)
  \end{equation*}  
\end{corollary}
Again, a similar result holds for optimality and $\epsilon$-optimality, in
case $D$ is finite.

In the following theorem, assume that $0<\epsilon<1/2$.

\begin{theorem}\label{thm:approx:extpts}
  Consider two decision problems $(\pspace,D,\credal,L)$ and
  $(\partition,D,\hat{\credal},\hat{L})$. Assume that $R_D>0$.
  If $L\sim_\epsilon\hat{L}$ and
  $\ext(\closure(\cohull(\credal)))\sim_\delta\hat{\credal}$ then, 
  for any $\gamma\ge 0$,
  \begin{align}
    \label{eq:thm:approx:extpts:1}
    {\max}^\gamma(\pspace,D,\credal,L)
    &\subseteq
    \bigcap_{\eta>0}{\max}^{\eta+\frac{\gamma}{1-2\epsilon}+2(\frac{\epsilon}{1-2\epsilon}+\delta)}%
    (\partition,D,\hat{\credal},\hat{L})
    \\
    \label{eq:thm:approx:extpts:2}
    {\max}^\gamma(\partition,D,\hat{\credal},\hat{L})
    &\subseteq
    \bigcap_{\eta>0}{\max}^{\eta+\gamma(1+2\epsilon)+2(\epsilon+\delta(1+2\epsilon))}(\pspace,D,\credal,L)
  \end{align}
\end{theorem}
\begin{proof}
  First, note that
  \begin{align*}
    {\max}^\gamma(\pspace,D,\credal,L)
    &=
    {\max}^\gamma(\pspace,D,\cohull(\credal),L)
    \\
    &\subseteq
    {\max}^{\gamma}(\pspace,D,\closure(\cohull(\credal)),L)
    \\
    \intertext{and by convexity of $\closure(\cohull(\credal))$ \cite[\S
      26.23]{1997:schechter} and the Krein-Milman theorem
      \cite[p.~74]{1975:holmes}, the closed convex hull of
      $\ext(\closure(\cohull(\credal)))$ is $\closure(\cohull(\credal))$, so}
    &=
    {\max}^{\gamma}(\pspace,D,\closure(\cohull(\ext(\closure(\cohull(\credal))))),L)
    \\
    \intertext{and now by Corollary~\ref{cor:credal:closure:maximality},}
    &=
    \cap_{\eta>0}{\max}^{\gamma+\eta}(\pspace,D,\cohull(\ext(\closure(\cohull(\credal)))),L)
    \\
    &=
    \cap_{\eta>0}{\max}^{\gamma+\eta}(\pspace,D,\ext(\closure(\cohull(\credal))),L)
  \end{align*}
  Now apply the same argument as in the proof of
  Theorem~\ref{thm:approx:opt} to recover
  Eq.~\eqref{eq:thm:approx:extpts:1}.

  To establish Eq.~\eqref{eq:thm:approx:extpts:2},
  again use the same argument as in the proof of
  Theorem~\ref{thm:approx:opt},
  \begin{align*}
    {\max}^\gamma(\partition,D,\hat{\credal},\hat{L})
    &\subseteq
    {\max}^{\gamma(1+2\epsilon)+2(\epsilon+\delta(1+2\epsilon))}(\pspace,D,\ext(\closure(\cohull(\credal))),L)
    \\
    &\subseteq
    {\max}^{\gamma(1+2\epsilon)+2(\epsilon+\delta(1+2\epsilon))}(\pspace,D,\closure(\cohull(\ext(\closure(\cohull(\credal))))),L)
    \\
    \intertext{and again by the Krein-Milman theorem
      \cite[p.~74]{1975:holmes}, the closed convex hull of
      $\ext(\closure(\cohull(\credal)))$ is $\closure(\cohull(\credal))$, so}
    &=
    {\max}^{\gamma(1+2\epsilon)+2(\epsilon+\delta(1+2\epsilon))}(\pspace,D,\closure(\cohull(\credal)),L)
    \\
    &=
    \bigcap_{\eta>0}{\max}^{\eta+\gamma(1+2\epsilon)+2(\epsilon+\delta(1+2\epsilon))}(\pspace,D,\cohull(\credal),L)
    \\
    &=
    \bigcap_{\eta>0}{\max}^{\eta+\gamma(1+2\epsilon)+2(\epsilon+\delta(1+2\epsilon))}(\pspace,D,\credal,L)
  \end{align*}
\end{proof}

Again, if we ignore higher order terms in $\gamma$, $\epsilon$, and $\delta$,
then the above theorem says that when moving from the original decision
problem to the approximate decision problem,
with relative error $\epsilon$ in gambles and relative error $\delta$
in probabilities, the relative error in maximality increases by
$2(\epsilon+\delta)$. Hence, for small $\epsilon$ and $\delta$
the following holds, up to a small error: 
if $L\sim_\epsilon\hat{L}$ and $\ext(\closure(\cohull(\credal)))\sim_\delta\hat{\credal}$, then
\begin{equation*}
    \max(\pspace,D,\credal,L)
    \subseteq
    {\max}^{2(\epsilon+\delta)}(\partition,D,\hat{\credal},\hat{L})
    \subseteq
    {\max}^{4(\epsilon+\delta)}(\pspace,D,\credal,L)
\end{equation*}
Again,
${\max}^{2(\epsilon+\delta)}(\partition,D,\hat{\credal},\hat{L})$ seems
a logical choice when calculating maximal decisions in practice.

\section{Conclusion and Remarks}
\label{sec:conclusion}

With this paper, I hope to have consolidated at least part of our
every day intuition when approximating decision problems involving
sets of probabilities, for instance when those problems have to be
solved by computer.

One result is quite depressing:
Lemma~\ref{lem:approx:gambles} and Lemma~\ref{lem:approx:probabilities} seem to tell us that
except in the simplest cases, any approximation will
need too many resources to be of any practical value, as demonstrated by
Table~\ref{tab:partition} and Table~\ref{tab:credal}.

Fortunately, not all is lost. 
If we resort to pairwise comparison, we may restrict ourselves to the
extreme points of the closure of the convex hull of the credal set,
which can be \emph{much} smaller than the original credal set.
Closing the credal set only has an arbitrary small effect on
maximality, and in part for this reason, it turns out that
approximating extreme points suffices when restricting to pairwise
preference.

I wish to emphasise that the bounds on the cardinalities of
the approximating partition and the approximating credal set are only
upper bounds under very weak assumptions. These bounds are only
attained in extreme situations. In many cases the credal set and the
loss function have additional structure which may allow for much lower
upper bounds.

In case the problem has sufficient structure, an alternative approach is to develop algorithms which do not need to
traverse the complete credal set (or an approximation thereof) to
compute the optimal solution. The imprecise Dirichlet model has
already been given considerable attention in this direction
\cite{2003:hutter:estimation}.

Obermeier and Augustin \cite{2007:obermeier::approx} have described a
method to approximate decision problems by applying Luce\~{n}os'
adaptive discretisation method to either all elements of the credal
set (so the partition varies with the distribution), or on a reference
distribution of that set. This type of approximation aims to preserve
the first $r$ moments of a distribution. Although precise convergence
results and bounds on the precision of this approximation have not yet
been proven, examples have shown that this method can yield good
results in practice.

Finally, another approach could consist of sampling elements from the credal set,
for instance through Monte-Carlo techniques, and solve a classical
decision problem for each of these elements. If the sample
$s$ from $\hat{\credal}$ is large enough, then---since $\bigcup_{P\in s}\opt(\partition,D,P,L)=\opt(\partition,D,s,L)$---hopefully
\begin{equation*}
  \opt(\partition,D,\credal,L)=\bigcup_{P\in s}\opt(\partition,D,P,L)
\end{equation*}
The question how large a sample we need to ensure convergence is
definitely worth further investigation.

\section*{Acknowledgements}

I am grateful to Teddy Seidenfeld for the many helpful discussions on
issues related to this paper, and also for encouraging me to extend my view on approximations to
choice functions. I thank Max Jensen for his help in characterising
the discretisation of the simplex in $\SetR^n$, presented in the
appendix.
I also thank all three referees for their constructive comments and useful suggestions which have improved the presentation of this paper.
The research reported in this paper has been supported in part by the
Belgian American Educational Foundation.

\appendix

\section{Discretisation Of The Standard Simplex In $\SetR^n$}

In this appendix a simple discretisation of $\Delta^n$, the standard simplex in $\SetR^n$,
is studied---these results are not new and are in fact related to well known notions
from combinatorics, in particular multisets \cite{1997:stanley}. The standard simplex $\Delta^n$ is defined as
\begin{equation*}
  \Delta^n=\left\{\vec{x}\in\SetR^n\colon \vec{x}\ge 0,\,|\vec{x}|_1=1\right\}
\end{equation*}
where $|\cdot|_1$ denotes the $1$-norm, i.e.\
$|\vec{x}|_1=\sum_{i=1}^n|x_i|$.

For any non-zero natural number $N$, let
$\Delta^n_N$ denote the following finite subset of $\Delta^n$:
\begin{equation*}
  \Delta^n_N=\left\{\vec{m}/N\colon \vec{m}\in\SetN^n,\,|\vec{m}|_1=N\right\}
\end{equation*}
(above, $\SetN$ is the set of natural numbers including $0$).

\begin{lemma}\label{lem:simplex:cardinality}
  The cardinality of $\Delta^n_N$ is $\binom{N+n-1}{N}$.
\end{lemma}
\begin{proof}
  There is an obvious one-to-one and onto correspondence between
  $\Delta^n_N$ and all multisets of cardinality $N$ with elements
  taken from $\{1,\dots,n\}$---for any $\vec{m}/N\in\Delta^n_N$,
  interpret $m_i$ as the multiplicity of $i$. The number of all such
  multisets is precisely $\binom{N+n-1}{N}$ (see Stanley
  \cite{1997:stanley}).
\end{proof}

\begin{lemma}\label{lem:simplex:approx}
  For every $\vec{x}$ in $\Delta^n$ there is a
  $\vec{y}$ in $\Delta^n_N$ such that
  \begin{equation*}
    |\vec{x}-\vec{y}|_1<n/N
  \end{equation*}
\end{lemma}
\begin{proof}
  For each $i\in\{1,\dots,n\}$, let $m_i$ be the unique natural number
  such that $x_i\in[m_i/N,(m_i+1)/N)$, or equivalently, let $m_i$ be
  the largest natural number such that $m_i/N\le x_i$. Define
  $M=\sum_{i=1}^n m_i$. Then, $M\le N< M+n$ since
  $M/N=|\vec{m}/N|_1\le|\vec{x}|_1=1$ and $(M+n)/N=|(\vec{m}+1)/N|_1>|\vec{x}|_1=1$. Define
  \begin{equation*}
    e_i=
    \begin{cases}
      1&\text{if }i\in\{1,\dots,N-M\} \\
      0&\text{if }i\in\{N-M+1,\dots,n\}
    \end{cases}
  \end{equation*}
  and let $\vec{y}=(\vec{m}+\vec{e})/N$. Note that
  $\vec{y}\in\Delta^n_N$ because
  $|\vec{y}|_1=|\vec{m}+\vec{e}|_1/N=(M+(N-M))/N=1$. Finally,
  \begin{equation*}
    |\vec{x}-\vec{y}|_1
    =
    \sum_{i=1}^{N-M}
    |x_i-\tfrac{m_i+1}{N}|
    +
    \sum_{i=N-M+1}^{n}
    |x_i-\tfrac{m_i}{N}|
    < n/N
  \end{equation*}
  as $|x_i-\frac{m_i+1}{N}|\le 1/N$ and $|x_i-\frac{m_i}{N}|< 1/N$.
\end{proof}

\bibliographystyle{amsplainurl}
\bibliography{all}

\end{document}